\documentclass[11pt]{article}
\usepackage{amsmath, amssymb, amsthm, algorithm, algorithmic, color}
\usepackage{enumerate, graphicx, hyperref}
\usepackage{cite}
\usepackage{wrapfig}
\usepackage{authblk}
\usepackage[margin=1in]{geometry}
\newcommand{\Real}{\mathbf{R}}
\DeclareMathOperator{\rank}{rank}

\newtheorem{lemma}{Lemma}
\newtheorem{theorem}{Theorem}
\newtheorem{corollary}{Corollary}

\title{Windows into Relational Events:\\ Data Structures for\\ Contiguous Subsequences of Edges}
\author[1]{Michael J. Bannister}
\author[2]{Christopher DuBois}
\author[1]{David Eppstein}
\author[1]{Padhraic Smyth}
\affil[1]{Department of Computer Science, University of California, Irvine}
\affil[2]{Department of Statistics, University of California, Irvine}
\begin{document}
\maketitle
\begin{abstract}
We consider the problem of analyzing social network data sets in which the edges of the network have timestamps, and we wish to analyze the subgraphs formed from edges in contiguous subintervals of these timestamps. We provide data structures for these problems that use near-linear preprocessing time, linear space, and sublogarithmic query time to handle queries that ask for the number of connected components, number of components that contain cycles, number of vertices whose degree equals or is at most some predetermined value, number of vertices that can be reached from a starting set of vertices by time-increasing paths, and related queries.
\end{abstract}

\thispagestyle{empty}
\newpage
\setcounter{page}{1}
\pagestyle{plain}

\section{Introduction}
The study of algorithms for social network analysis has so far been concentrated primarily on computations involving graphs that are relatively static: either a fixed graph is given as input to algorithms for problems such as the computation of centrality~\cite{EppWan-SODA-01,Bra-JMS-01,OkaCheLi-FAW-08,Kin-08}, or the input graph is assumed to change gradually by insertions and deletions of vertices and edges, and these changes can be handled efficiently by dynamic graph algorithms~\cite{EppGalIta-ATCH-99,EppSpi-WADS-09,EppGooStr-COCOA-10}. These input models work well for networks that describe long-term ties such as friendship or supervisorial relations between people, and they often match the information provided by online service providers such as Facebook. However, there is a second type of social network data set, known variously as \emph{relational event data}~\cite{But-SM-08}, \emph{dyadic event data}~\cite{BraLerSni-ASONAM-09}, \emph{longitudinal network data}~\cite{GolZheFie-FTML-09,Noo-SN-11},  \emph{contact sequences}~\cite{Hol-PRE-05}, or \emph{time-ordered networks}~\cite{Moo-SF-02}, on which we would also like to perform efficient computations. This type of data models communication events between pairs of people rather than long-term ties; for instance, each datum in a data set might consist of the identities of the sender and recipient of a single email message along with its timestamp.\footnote{For example, Kossinets and Watts describe an email data set of this type with approximately 7 million email messages, sent by approximately 30,000 people within a single university campus over the course of a year~\cite{KosWat-AJS-09}.} At any instant of sampled time there is no graph, only a single edge.

It is only by grouping together multiple events over sliding windows of time that we can form a network from this type of relational event data~\cite{CorPreVol-JCGS-03,MooMcFBen-AJS-05,KosWat-AJS-09}. If a fixed window size is chosen, the sequence of time windows can be modeled by the insertion of an edge when it enters the window, and deletion when it leaves the window; these types of change are familiar in the analysis of algorithms as the basis for many dynamic graph algorithms~\cite{EppGalIta-ATCH-99}. However, the dynamic graph model is too restrictive for our purposes, because it may not be obvious what length of time window to use. Long windows present a relatively static view of the data, aggregating it into a single graph but losing all dynamic time information. Short windows capture the dynamics better but may be too sparse to see the entire pattern of connections; indeed, for very short windows, the subset of data within any window may have few or no vertices with degree higher than one. Thus, it may be useful to perform exploratory data analysis by testing different sizes of window to find the one that best balances connectivity with dynamics, or to study the same data set with multiple window sizes to show how its behavior varies with the time scale.

In this paper we provide for the first time an algorithmic model for the analysis of relational event data with windows chosen dynamically rather than a priori, and we also develop fundamental data structures that can perform this analysis efficiently. In our model, the input to a relational event data analysis problem is formulated as a sequence of (directed or undirected) edges, and we define a \emph{slice} of the data to be the graph formed by a contiguous subsequence of the input. The data structures that we describe  represent the entire relational data set using only linear space, can be constructed in near-linear time, and support queries that ask for statistical information about arbitrary slices of the data in sublogarithmic time per query. The graph properties that our data structures can handle include many quantities already studied (for static networks) by social networking researchers, including the following:
\begin{itemize}
\item We can count the number of connected components of a slice, the number of nontrivial components (with more than one vertex), the number of \emph{loopy components} (components that contain at least one cycle), the average size of a connected component,  the average size of a nontrivial connected component, and the number of \emph{loopy edges} (edges that close a cycle).  Such queries concern the connectivity structure of the network, providing insight into diffusion processes and the robustness of these processes to intervention.  In sexual networks, for example, a scarcity of short cycles implies the absence of a dense core, a fact with implications for the study and treatment of HIV \cite{BeaMooSto-AJS-04, PotPhiPlu-STI-02} and Gonorrhoea~\cite{DeSinWon-STI-04, PotMutRot-STI-02}. In such networks loopy edges represent reinfection events, which fuel the growth phase of an outbreak~\cite{PotMutRot-STI-02}.
\item We can count the number of isolated vertices in a slice, the number of isolated edges, the number of edges that have $d$ neighboring edges for any constant~$d$, and the number of vertices that have exactly or at most $d$ neighboring edges or vertices for any predetermined (but not necessarily constant)~$d$. These parameters provide access to the \emph{degree distribution} of a network, which has long been recognized as important in social network analysis; for instance, Seidman~\cite{Sei-SN-83} emphasizes the importance of distinguishing between networks with uniform degrees from networks in which there are a few high degree vertices and many low degree vertices.
\item We can count the number of repeated edges within a slice, the number of edges with multiplicity exactly or at most~$\mu$, and the \emph{reciprocity} (relative proportion of pairs of vertices connected by directed edges in both directions to pairs connected only in one direction). In social networks where directed edges represent communication, reciprocity has been used to quantify the amount of social interaction versus broadcast communication~\cite{Gong11, Diego04}.
\item We can count the number of vertices that can be reached from a predetermined set of starting vertices, via (directed or undirected) paths in which the timestamps of the edges are monotonically increasing. These paths have also been called \emph{journeys}~\cite{BuiFerJar-WiOpt-03} or \emph{diffusion paths}~\cite{Moo-SF-02}; they are the possible transmission routes of information or contagion through the network, and the number of reachable nodes has been studied (for the aggregate graph rather than for slices) by Holme~\cite{Hol-PRE-05}. We can also count the number of vertices that can be reached by paths of this type with bounded hop count.
\item We can count the number of \emph{triad closure events} in which an edge belongs to at least one triangle formed by it and earlier edges within the slice. Triads and related transitivity properties such as the clustering coefficient have long been recognized as important in the structure of social networks~\cite{Gra-AJS-73,Rap-BMB-53} and the number of triad closure events, like monotonic reachability,  also incorporates time-dependence in its definition. Our data structure for counting them involves somewhat slower preprocessing, comparable to the time to find or count triangles in a static graph.
\end{itemize}

\begin{wrapfigure}[14]{r}{0.28\textwidth}
\vspace{-1.5em}
\centering
\includegraphics[width=0.28\textwidth]{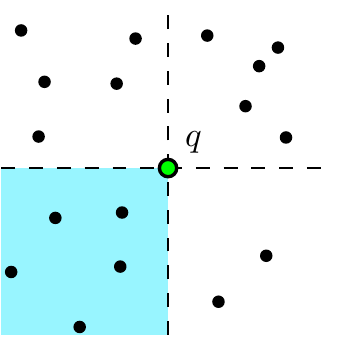}
\vspace{-2em}
\caption{Example dominance query where the points in the shaded region are counted.}
\label{fig:dom-q}
\vspace{-0em}
\end{wrapfigure}
We show how to reduce each of these problems to two-dimensional \emph{dominance counting} queries on point sets derived from the input network, with $O(1)$ points per network edge. In the variant of two-dimensional dominance counting that we use, a data set consists of a set of $n$ two-dimensional points with integer coordinates $(x,y)$ satisfying $1\le y\le x\le n$, and each query must determine the number of points dominated by a query point $q = (x_0, y_0)$, i.e., the number of points $(x, y)$ with $x \leq x_0$ and $y \leq y_0$. Through a slight abuse of terminology we will use the term dominance counting when counting the number of points contained in any one of the four quadrants created by placing the query point $q$ in the plane.  As we describe, many of the problems listed above can be reduced to dominance counting through a common unification in terms of the rank function of matroids,  generalizing a data transformation for one-dimensional colored range counting given by Gupta et al.~\cite{ColorRangeQueries}.
For the remaining problems (including isolated edges and monotonic reachability) our reduction instead passes through another two-dimensional range searching problem, rectangle stabbing. Combining our reductions with known dominance counting data structures~\cite{JaJMorShi-ISAAC-04} would allow us to solve these problems in linear space and query time $O(\log m/\log\log m)$, where $m$ is the number of edges in the input network; in an appendix we provide a refined dominance counting structure that improves the query time to $O(\log w/\log\log m)$ where $w$ is the number of edges in the queried slice. We also outline an alternative solution based on path-copying persistence~\cite{DSPersistent} and balanced binary trees that we expect to be more suitable for implementation; it uses $O(m\log m)$ space and gives query time $O(\log w)$.

In another appendix, we adapt a lower bound of Mihai P{\v a}tra{\c s}cu for two-dimensional range counting~\cite{Pat-STOC-07} to the problems studied here, using reductions in the other direction from point sets to networks. We show that, in the cell probe model, with query time measured as a function only of the input size (rather than of the window size, as in our upper bounds) all data structures that use space $O(m\log^{O(1)} m)$ must take $\Omega(\log m/\log\log m)$ query time to solve many of the queries considered here.

\section{Problem formulation}
Define a \emph{relational event graph} $G$ to be a fixed set of vertices $V$ together with a sequence of edges (or relational events) $E = \{e_k \mid 0 \leq k < m\}$ between pairs of vertices. The graph is undirected if the pairs are unordered, and directed if the pairs are ordered. The pairs in the sequence are not required to be distinct from each other.
Given a relational event graph $G$ we define the \emph{slice multigraph} $G_{i,j}$ to be the multigraph with vertices $V$ and edges $\{ e_k\mid i\le k\le j\}$.

We assume that the entire relational event graph $G$ is given to us as input.  Our task is to construct a data structure from $G$ that will allow us to compute the properties of its slices $G_{i,j}$, for a query pair of indices $i,j$. To avoid trivial solutions, queries in such a data structure should take less time than the $\Omega(j-i)$ of an algorithm that constructs the slice multigraph and applies a static graph algorithm to it, and the data structure should use less space than the $\Omega(m^2)$ of an algorithm that precomputes and stores the answers to all possible queries.

\section{Matroid rank}
We will turn many of our queries into a matroid rank problem. Recall~\cite{Lawler01, Welsh10} that a \emph{matroid} over a set $S$ is a collection of subsets of $S$ called \emph{independent sets} 
obeying the following three properties:
\begin{itemize}
\item The empty set is independent.
\item Every subset of an independent set is independent.
\item If $U$ and $V$ are independent sets and $U$ is larger than~$V$, then there exists $u\in U$ such that $V \cup \{u\}$ is independent.
\end{itemize}
The \emph{rank} of a set $E \subseteq S$ is the size of the largest independent subset of $E$, and a \emph{circuit} is a minimal dependent subset (i.e., a set whose proper subsets are all independent).

We will define matroids over sequences $E=\{e_k \mid 0 \leq k < m\}$ of elements (usually the edges of our input graph); we form slices $E_{i,j} = \{e_k\mid i\le k\le j\}$ from contiguous subsequences of this sequence. The rank of $E_{i,j}$ will then be useful for computing the numerical graph quantities we wish to compute; for instance, we will use the graphic matroid (whose rank is the number of edges in a spanning forest) to determine the numbers of connected components and loopy edges in a slice. In order to calculate the rank of a slice, we define the \emph{independence time} $\tau(e_k)$ for an element $e_k$ to be the smallest index $i$ such that $\rank(E_{i,k})>\rank(E_{i,k-1})$, or $-1$ if the inequality holds for all $i$. A straightforward induction on $j-i$ shows that the rank of every slice $E_{i,j}$ equals the number of matroid elements in the slice whose independence is before the interior of the slice.

Our meta-algorithm for precomputing $\tau$ (needing details to be filled in for  specific matroids) assigns a weight to each element, equal to its index. It then incrementally considers the elements in sequence order, maintaining as it does a maximum-weight basis of the set of elements considered so far. When adding element $e_k$ to the basis would cause it to remain independent, the augmented set becomes the new basis and in this case we set $\tau(e_k)=-1$. However, when the previous basis and the new element together contain a circuit (necessarily a unique circuit), we form the new basis by removing the lightest element from this circuit, and adding $e_k$ in its place; in this case, $\tau(e_k)$ is one more than the index of the removed element. Later, when we discuss specific matroids, we will describe how to quickly identify the circuit containing the new element and the lightest element of this circuit.

\begin{lemma}\label{lem:general-matroid}
The elements of $E$ can be mapped to points in $\Real^2$ such that the rank of $E_{i,j}$ can be determined by a dominance counting query. The time needed for this mapping is the same as the time to compute $\tau$ for all elements in $E$.
\end{lemma}
\begin{proof}
We map each $e_k$ to $(k, \tau(e_k))$. To determine the rank of $E_{i,j}$ we count the number of elements whose independence time is in the slice's range of indices, i.e.,  $i \leq k \leq j$ and $i \leq \tau(e_k)$. This three sided query can be reduced to the dominance counting query $k \leq j$ and $i \leq \tau(e_k)$, as it is not possible to have $k < \tau(e_k)$. We then take the complement to count the edges whose independence time is before the interval.
\end{proof}

\section{Counting vertices by degree and edges by multiplicity}
In this section we use partition matroids (a standard type of matroid, defined below) to determine the number of vertices of bounded degree, the number of vertices of a specific degree, the number of edges of bounded multiplicity, the number of edges of a given multiplicity, and the reciprocity of a slice in a relational event graph. Our techniques can also solve the colored range counting problem considered in~\cite{ColorRangeQueries}, using colors to define the partition, and our Lemma~\ref{lem:vertex-degree} generalizes their data transformation approach for colored range counting.

In general, a partition matroid is defined over a set $S$ that has been partitioned into a family of disjoint subsets $P_i$ for $1 \leq i \leq p$, each of which is associated with a numeric parameter $k_i$. A subset $A$ of $S$ is defined to be independent in the matroid if $A \cap P_i$ has at most $k_i$ elements for each $1 \leq i \leq p$. Each circuit of this matroid is a subset of exactly $k_i+1$ elements of one of the sets $P_i$. It is straightforward to verify that the independence system defined in this way satisfies the axioms of a matroid~\cite{Welsh10}.

Given a relational event graph $G$ with edge sequence $E = \{ e_i \}$ and a parameter $k$ we consider a partition matroid whose elements are the set of half-edges $(u_i, 2i)$ and $(v_i, 2i + 1)$ for each edge $e_i = \{u_i, v_i\}$ in $E$. In this matroid we define a set $S$ to be independent if each vertex of $G$ appears at most $k$ times in the first component of a half-edge, where $k$ is a fixed parameter. That is, there is one set $P_u$ for each vertex $u$, containing all the half-edges $(u,v)$, and the corresponding partition matroid parameter is $k_u=k$. The rank of this partition matroid is given by $\rank_k(S) = \sum_{v \in V} \max(\deg(v), k).$

\begin{lemma}\label{lem:vertex-degree}
We can compute $\tau$ for the partition matroid described above in linear time and space.
\end{lemma}
\begin{proof}
We process the half-edges in index order, adding them to a dictionary that associates each vertex with a $k$-element queue of insertion times. When inserting a new half-edge, if its endpoint's queue is full, we dequeue the top half-edge and record one more than its index for $\tau$ of the current half-edge; otherwise we record $-1$ for $\tau$ of the edge. Then, regardless of whether the queue was full, we add the half-edge to the queue. By storing the queues as linked lists this can be done in linear space and time.
\end{proof}

\begin{theorem}\label{thm:vertex-degree}
Given a relational event graph $G$ the problems of determining the number of isolated vertices, vertices of a given degree~$d$, and vertices of bounded degree in $G_{i,j}$ can be reduced to dominance counting in linear time and space.
\end{theorem}

\begin{proof}
We use our solution to the matroid rank problem (Lemma~\ref{lem:general-matroid}) to create two data structures for the partition matroid with $k = d$ and with $k = d+1$. Then by performing two dominance counting queries we can compute
\[
\rank_{d+1}(G_{i,j}) - \rank_{d}(G_{i,j}) = \sum_v \max(\deg(v), d+1) - \max(\deg(v), d)
\]
which is equal to the number of vertices of degree greater than $d$. With the ability to count the number of vertices of degree greater than $d$ we can easily compute the queries stated in the theorem by inclusion-exclusion.
\end{proof}

We define the \emph{multiplicity} of an edge in a directed or undirected relational event graph to be the number of other edges that have the same two endpoints (as an ordered or unordered pair, respectively).
We can count the distinct edges, the edges that have a given multiplicity, or the edges that have bounded multiplicity, using a partition matroid whose elements are the edges, and whose partitions group together edges that have the same ordered or unordered pair of endpoints.

\begin{theorem}
Given a (directed) relational event graph $G$ the problems of determining the number of distinct directed or undirected edges, edges of bounded multiplicity, reciprocated edges, and the reciprocity in  $G_{i,j}$ can be reduced to dominance counting in linear time and space.
\end{theorem}

\begin{proof}
The number of edges with given multiplicity can be counted using ranks in a partition matroid, as in Theorem~\ref{thm:vertex-degree}.

To determine the number of reciprocated edges, we count the number of distinct edges in two different ways, interpreting the same graph once as a directed graph and a second time as an undirected graph. Reciprocated edges are counted twice as distinct directed edges but only once as undirected, and unreciprocated edges are counted once either way, so the number of reciprocated edges is the difference between the numbers of distinct directed and undirected edges. The reciprocity is then the ratio of reciprocated edges to all edges.
\end{proof}

\section{Counting connected components}
To count the number of connected components in a graph we use the \emph{graphic matroid}~\cite{Welsh10}, another standard type of matroid. The graphic matroid of an undirected graph $G$ has the edges of $G$ as its elements; a set of edges is independent in the graphic matroid if it forms a forest. A circuit in the graphic matroid is a simple cycle in $G$, and the rank of the graphic matroid on a set of edges is the number of vertices minus the number of connected components.

To count connected components that contain cycles we use the \emph{bicycle matroid} (or \emph{bicircular matroid})~\cite{Matthews77}, a somewhat less-well-known matroid that also has the edges of a graph as its elements. In the bicycle matroid, a set of edges is independent if it forms a \emph{pseudoforest}, a graph that has at most one cycle per connected component or equivalently a graph in which each subgraph has at most as many edges as vertices. The rank of the bicycle matroid on a set of edges is the number of vertices minus the number of tree components.

To precompute $\tau$ for both of these matroids, we use \emph{linking and cutting trees}~\cite{LinkCutArt, LinkCutBook}. These are data structures that may be used to represent a rooted forest, subject to updates that either insert edges (if the result of the insertion would still be a forest) or delete them. Cutting and linking trees also allow operations to look up the root of the tree containing a query vertex or the lightest edge on any path. Both updates and queries take logarithmic time per operation, and the overall data structure uses linear space.

\begin{lemma}\label{lem:connected-component}
We can precompute $\tau$ for the graphic matroid in $O(m\log n)$ time and linear space.
\end{lemma}
\begin{proof}
We store the vertices of $G$ into a linking and cutting tree, and then process the edges in increasing order. When adding an edge $e_i = \{u_i, v_i\}$ to the forest we check if it creates a cycle. If so, we find the lightest edge on the path from $u_i$ to $v_i$, record one more than its index as $\tau(e_i)$, remove the light edge from the forest, and add $e_i$ to the forest. If adding $e_i$ does not create a cycle, then we add it to the forest and record $-1$ for $\tau(e_i)$. For each edge we do $O(\log n)$ work, for a total processing time of $O(m\log n)$.
\end{proof}

\begin{theorem}\label{thm:connected-component}
Given a relational event graph $G$ the problems of determining the number of connected components, nontrivial connected components, average size of a connected component,  average size of a nontrivial connected component, and the number of loopy edges in $G_{i,j}$ can be reduced to dominance counting in $O(m\log n)$ time and linear~space.
\end{theorem}

\begin{proof}
The number of connected components follows from the matroid rank problem (Lemma~\ref{lem:general-matroid}) together with Lemma~\ref{lem:connected-component}. For nontrivial components we use Theorem~\ref{thm:vertex-degree} to count isolated vertices and subtract this value from the number of forests. The average component sizes can then be computed easily. To count the number loopy edges we observe that the number of loopy edges equals the total number of edges minus the number of edges in a spanning forest, i.e., it is the number of edges in the slice minus the graphic matroid rank of the slice.
\end{proof}

When computing $\tau$ for the bicycle matroid we will need to dynamically maintain a pseudoforest. To do this we augment the linking and cutting tree with a dictionary whose keys are the tree roots and whose associated values are the lightest edges in the cycles of the corresponding pseudotrees (or null for tree components). Thus, the linking and cutting tree always stores the maximum spanning forest, and the dictionary holds the missing edges of each pseudotree.

\begin{lemma}\label{lem:loopy-component}
We can precompute $\tau$ for the bicycle matroid in $O(m\log n)$ time and linear space.
\end{lemma}
\begin{proof}
When adding a edge $e_k = \{u_k, v_k\}$ we consider five possible cases: (1) two trees are joined, (2) a cycle is created in a tree,  (3) a tree and a pseudotree are joined, (4) a second cycle is formed in a pseudotree, (5) two pseudotrees are joined.

In cases (1), (2) and (3) we are left with a pseudoforest so we record $-1$ for $\tau(e_k)$. In case (2) we remove the lightest edge on the cycle formed by $e_k$ from the linking and cutting tree, and place it in the dictionary; in all cases we add $e_k$ to the linking and cutting tree. In case (3) we update the key for the lightest edge in the pseudotree with its new root, if the root changes.

In cases (4) and (5) we find and discard the lightest edge in the union of the two cycles and the path (if it exists) between them, 
either returning to a component with one cycle or splitting it into two components each with a cycle. We update the cutting and linking tree and dictionary, and record one more than the index of the discarded edge as $\tau(e_k)$.

Since we only added $O(n)$ space and $O(m\log n)$ time to the procedure in Lemma~\ref{lem:connected-component} we have the same space and time bounds.
\end{proof}

\begin{theorem}
Given a relational event graph $G$ the problems of determining the number of loopy components, tree components, and nontrivial tree components in $G_{i,j}$ can be reduced to dominance counting in $O(m\log n)$ time and linear space.
\end{theorem}
\begin{proof}
The number of trees follows from the matroid rank problem (Lemma~\ref{lem:general-matroid}) and Lemma~\ref{lem:loopy-component}. For loopy components we also build the data structure in Theorem~\ref{thm:connected-component} and subtract the number of trees from the number of connected components. For nontrivial trees we build the data structure in Theorem~\ref{thm:vertex-degree} and subtract the number of isolated vertices from the number of trees.
\end{proof}

\section{Counting edge neighbors}
In this section we compute the number of edges that have a given or bounded number of neighboring edges. This does not seem to be an instance of the matroid rank problem. Instead, we reduce it to a rectangle stabbing problem.

For an edge $e_k$ we define a \emph{past neighbor} to be an edge $e_i$ sharing at least one vertex with $e_k$  and having $i < k$, and we define a \emph{future neighbor} to be an edge $e_j$ sharing at least one vertex with $e_k$ and having $k < j$. Let $\pi_r(e_k)$ denote the least $i$ such that $e_k$ has $r$ neighbors in $G_{i,k}$ by $\pi_r(e_k)$, and let $\phi_r(e_k)$ denote the greatest $j$ such that $e_k$ has $s$ neighbors in~$G_{k,j}$.

\begin{lemma}
We can precompute $\pi_r$ and $\phi_s$ in $O((r+s)m)$ time and linear space.
\end{lemma}
\begin{proof}
First we precompute $\pi_r$ by processing the edges (as half-edges) in index order, adding them to a dictionary structure, as in Lemma~\ref{lem:vertex-degree}, indexed by the vertex and storing the insertion times in $r$-sized queues. When inserting a new edge $e_k = \{u_k, v_k\}$ we consider the queues for both $u_k$ and $v_k$. If the sum of the sizes of the two queues is less than $r$ then we record $-1$ for $\pi_r(e_k)$. Otherwise we iterate through the queue to find the least $i$ such that there are exactly $r$ neighbors with index greater than $i$ and record this as $\pi_r(e_k)$. Finally, we add the two half-edges to their respective queues, dropping the top half-edge if the queues overflow. This process is done in $O(rm)$ time, $O(r)$ to iterate through the queues, and $O(n + m)$ space. To compute $\phi_s$ we repeat the process in reverse, which takes $O(sm)$ time in $O(n+m)$ space.
\end{proof}

\begin{lemma}\label{lem:rectangle-stabbing}
Given a set of rectangles, the problem of determining which rectangles are stabbed by (enclosing) a query point can be reduced to a constant number of dominance counting queries.
\end{lemma}

\begin{wrapfigure}[13]{r}{0.26\textwidth}
\vspace{-1em}
\centering
\includegraphics[width=0.23\textwidth]{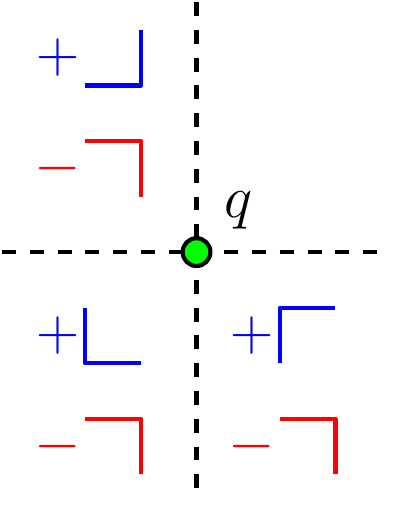}
\vspace{-1em}
\caption{Turning stabbing into dominance.}
\label{fig:rectangle-stab}
\vspace{-2em}
\end{wrapfigure}
\noindent\emph{Proof.}
The number of rectangles stabbed by point $q$ can be reduced to a linear combination of the counts of rectangle corners belonging to six different combinations of  corner type and apex-$q$ quadrant, using an inclusion-exclusion relation that seems to be folklore.
Figure~\ref{fig:rectangle-stab} provides an illustration: if we add $+1$ for each rectangle whose geometric relationship to $q$ is indicated by the blue L-shapes, and $-1$ for each rectangle whose relation to $q$ is indicated by the red L-shapes, then each rectangle containing $q$ adds a total of $+1$ to this sum (only for its lower left corner) while each other rectangle adds zero (either with two corners that cancel each other, or no corners). Therefore the total sum equals the number of rectangles stabbed by $q$.

Therefore, to answer rectangle stabbing queries, we may build three (signed) dominance counting data structures,  one for each of the nonempty quadrants in the figure, giving us the contributions from each quadrant. \qed

\begin{theorem}\label{thm:future-past}
Given a relational event graph $G$ the problem of determining the number of edges with at most $r$ past neighbors and at most $s$ future neighbors in $G_{i,j}$ can be reduced to dominance counting in $O((r+s)m)$ time and linear space.
\end{theorem}
\begin{proof}
An edge $e_k$ has at most $r$ past neighbors and at most $s$ future neighbors in $G_{i,j}$ (and is in $G_{i,j}$) precisely when $\pi(e_k) \leq i \leq k \leq j \leq \phi(e_k)$.
If we view $(i, j)$ as a point in $\Real^2$, then this happens when the rectangle $[\pi_r(e_k), k] \times [k, \phi_s(e_k)]$ encloses the point $(i, j)$, which reduces to dominance counting by Lemma~\ref{lem:rectangle-stabbing}.
\end{proof}

\begin{corollary}
Given a relational event graph $G$ the problem of determining the number of edges with $r$ past neighbors and $s$ future neighbors, the number of isolated edges, and the number of edges with $k$ neighbors in $G_{i,j}$ can be reduced to dominance counting in $O((r+s)m)$ time and linear space, except for the number of edges with $k$ neighbors which takes $O(km)$ time and space.
\end{corollary}
\begin{proof}
For edges with past and future neighbors we use Theorem~\ref{thm:future-past} to compute the four data structures that compute $N_{\le r', \le s'}$ (the number of edges with past and future edges bounded by $r'$ and $s'$) for all combinations of $r' \in\{ r, r-1\}$ and $s' \in\{ s,s-1\}$. Then to compute the number of edges with exactly $r$ past neighbors and $s$ future neighbors we use inclusion-exclusion:
\[
N_{=r,=s}=N_{\le r, \le s} - N_{\le r, \le s-1} - N_{\le r-1, \le s} + N_{\le r-1, \le s-1}.
\]
To count isolated edges we set $r=s=0$. To count edges with exactly $k$ neighbors we sum over the edges with $r$ past and $s$ future neighbors for all combinations of $r$ and $s$ satisfying $r+s = k$.
\end{proof}

\section{Determining influence}
In this section we designate a fixed set of vertices as \emph{influential vertices} and seek to find the number of \emph{influenced vertices}, where vertex $v$ is influenced if there is a path of index-increasing edges from a influential vertex to~$v$. Such a path will be called a \emph{path of influence}. If we think of the edges as communication events, then this models the flow of information from the influential vertices. Motivated by the degradation of information as it is relayed we also consider the number of \emph{$h$-influenced} vertices, i.e., vertices that are on a path of influence with less than $h$ edges. This query also does not appear to be an instance of the general matroid slice problem.

For each edge insertion $e_k = (u_k, v_k)$ we define $\iota(e_k)$ to be the greatest $i$ such that $v_k$ is influenced in $G_{i, k}$, and $\lambda(e_k)$ the least $j$ such that $v_k$ is influenced in $G_{k, j}$.

\begin{lemma}\label{lem:influence}
The values of $\iota$ and $\lambda$ can be computed in linear time and space.
\end{lemma}
\begin{proof}
We consider the edges $e_k = (u_k, v_k)$ in sequence order, setting $\iota(e_k)$ and $\lambda(e_k)$ as we do.

The edge $e_k = (u_k, v_k)$ is on a path on influence only if either $u_k$ is a influential vertex or an influenced vertex ($\iota(u_k)$ is set). If $u_k$ is an influential vertex, then $\iota_{k}(v_k) = k$. Otherwise, we set $\iota(v_k)$ to $\iota(u_k)$ when $\iota(v_k) < \iota(u_k)$ or $\iota(v_k)$ is unset, and do nothing when $\iota(v_k) \geq \iota(u_k)$. The computation of $\lambda$ is similar.
\end{proof}

\begin{theorem}\label{thm:influence}
Given a relational event graph $G$ the problem of determining the number of influenced vertices in the slice $G_{i,j}$ can be reduced to dominance counting in linear time and space.
\end{theorem}
\begin{proof}
We will count the number of influenced vertices in $G_{i,j}$ by counting the number of destination vertices of edges in $G_{i,j}$ that are not influenced (counted with multiplicity) and then taking the complement. A vertex $v$ cannot be influenced in $G_{i,j}$ unless there is an edges $e_k = (u_k, v_k)$ in $G_{i,j}$ with $v_k = v$, i.e., $i \leq k \leq j$. Now $v_k$ is not influenced in $G_{i,j}$ whenever $\iota(e_k) < i \leq k \leq j < \lambda(e_k)$, i.e., when the point $(i,j)$ is in the rectangle $(\iota(e_k), k] \times [k, \lambda(e_k))$. Now that the problem is reduced to rectangle stabbing we use Lemma~\ref{lem:rectangle-stabbing}.
\end{proof}

\begin{theorem}
Given a relational event graph and a predetermined value $h$ the problem of determining the number of $h$-influenced vertices in the slice $G_{i,j}$ can be reduced to dominance queries in $O(hm)$ time and linear space.
\end{theorem}
\begin{proof}
We modify the argument in Lemma~\ref{lem:influence} and Theorem~\ref{thm:influence} to keep track of $\iota$ and $\lambda$ for $k$-influence for each vertex and for each choice of  $k\le h$ instead of just for influence. It takes $O(h)$ time per edge to update these times of $k$-influence.
\end{proof}

\section{Counting triad closure events}

Define a \emph{triad closure event} in an undirected relational event graph to be an edge $e_k$ within a given slice $G_{i,j}$ such that $e$ is the final edge of at least one triangle; that is, such that the other two edges of the triangle also belong to the same slice but are earlier in the sequence of edges than $e_k$.
To count these events we define $\Delta(e_k)$ to be the smallest index $d$ such that $e_k$ does not belong to a triangle in $G_{d,k}$. Then, the number of triadic closure events for slice $G_{i,j}$ is exactly the number of edges $e_k$ satisfying $i<\Delta(e_k)<k\le j$, something that can be counted with the same mapping to $\Real^2$ and dominance query in Lemma~\ref{lem:general-matroid}. The difficulty, for this problem, is in the preprocessing: how do we compute $\Delta(e_k)$ efficiently, for all edges~$e_k$?

To solve this problem, we adapt a data structure of Eppstein and Spiro~\cite{EppSpi-WADS-09} for counting triangles in a dynamic graph. This data structure is based on the concept of the \emph{$h$-index} of the graph, the largest number $h$ such that the graph contains at least $h$ vertices of degree at least $h$; all graphs with $m$ edges satisfy $h=O(\sqrt m)$.  Eppstein and Spiro maintain a slowly-changing partition of the graph vertices into two subsets $H$ and $L$, where $H$ contains $O(h)$ vertices and where every vertex in $L$ has degree $O(h)$. We simplify this by computing the $h$-index of the aggregate graph and partitioning its vertices into static subsets $H$ and $L$, where $|H|\le h$ and where every vertex in $L$ has degree at most~$h$.

Next, we loop through the edges in sequence order, maintaining as we do two hash tables $E$ and $P$ indexed by pairs of vertices. The first of these two tables, $E[u,v]$, stores the most recent edge with those two endpoints (if such an edge has already been encountered in the edge sequence). The second table, $P[u,v]$ stores the two-edge path from $u$ to $v$ via a third node $w\in L$ that maximizes the index of the earlier of the two edges $(u,w)$ and $(w,v)$, if  such a path exists and $G$ has an edge~$(u,v)$. We also maintain an adjacency list for each vertex, listing the vertices connected to it by edges that have already been encountered.

From this information, we can compute $\Delta(e_k)$ in time $O(h)$: let $u$ and $v$ be the endpoints of $e_k$, look up in $P[u,v]$ the best path through a vertex in $L$, and find the best path through a vertex in $H$ by testing all $h$ choices for this vertex using $E$ to test each choice in constant time. Once $\Delta(e_k)$ has been computed, we may also update $E$ and the adjacency lists in constant time. To update $P$, for each endpoint $v$ of $e_k$ that belongs to $L$, loop through each neighbor $w$ of $v$, find the two-edge path combining $e_k$ and $E[v,w]$, and use this path to update $P[u,w]$ where $u$ is the other endpoint of $e_k$. This update process takes constant time per neighbor, and there are at most $h$ neighbors, so again the time is $O(h)$.

\begin{theorem}
Given an undirected relational event graph $G$ the problem of determining the number of triad closured in the slice $G_{i,j}$ can be reduced to dominance counting in $O(hm)$ time and linear space.
\end{theorem}
\begin{proof}
We perform the preprocessing steps described above to compute $\Delta(e_k)$ for each edge $e_k$, in total time $O(hm)$, and then use the same persistent finger tree structure described in the matroid rank data structure (Lemma~\ref{lem:general-matroid}), using $\Delta$ in place of the similar index $\tau(e_k)$ of the matroid rank data structure.
\end{proof} 

\section{Conclusions}
We have described data structures for many counting problems on slices of relational event data. Our analysis separates preprocessing from queries,  but many of our data structures preprocess the data in sequence order, allowing queries to be interleaved with the addition of new data to the end of the sequence.

Many interesting social network parameters remain to be addressed, including the clustering coefficient, the $h$-index, the number of vertices reachable via non-monotonic paths, and the size of the largest connected component. In addition, several of the parameters for the statistics we compute (such as the hop count and  influential vertices in our influence-counting structure) must be determined at preprocessing time, and it would be of interest to develop more flexible structures that can delay the choice of these parameters until query time. Thus, although we have shown many interesting graph statistics to be computable efficiently in our model, much more remains to be done.

\newpage

\subsection*{Acknowledgements.}
This research was supported in part by the National Science Foundation under grant 0830403, and by the Office of Naval Research under MURI grant N00014-08-1-1015.

{\raggedright
\bibliographystyle{abuser}
\bibliography{win-rel-events}}

\newpage
\appendix
\section{Window-sensitive dominance counting}

Suppose we are given as input a set $S$ of $n$ points, with integer coordinates in the range from $1$ to $n$; we wish to answer \emph{dominance counting queries}, where a query specifies a point $(x,y)$ and must count the number of points $(x',y')\in S$ with $x'\le x$ and $y'\le y$.
JaJa, Mortensen and Shi~\cite{JaJMorShi-ISAAC-04} provide a data structure for this problem, in the word RAM computation model, that uses linear space and achieves $O(\log n/\log\log n)$ query time. More precisely, they show (in their Lemma 5) that in a model of computation in which each word contains at least $\log n$ bits of information and in which tables of size $n$ may be precomputed,
then it is possible to represent sets of $m$ points in space $O(m)$ and achieve query time $O(\log m/\log\log n)$. (Gupta et~al.{} make an additional assumption, that the points of their data set have distinct coordinates, but this can be achieved with no loss of generality and with no change to their space or query time bounds by sorting the points by their coordinate values and replacing the coordinates by indices into the sorted order.)
Applying this structure directly to the point sets generated from our reductions would give us query time $O(\log m/\log\log m)$ and space $O(m)$, where $m$ is the number of edges in the given relational event graph. Instead, we show that it is possible to achieve slightly faster query time, $O(\log w/\log\log m)$, where $w$ is the number of edges in the query slice.

The key observations needed for this improvement are the following:
\begin{itemize}
\item
All of the points $(x_i,y_i)$ in the point sets generated by our reductions satisfy $x_i\ge y_i$; that is, they lie below the main diagonal $x=y$ of the $n\times n$ square forming the bounding box of the points.
In the matroid rank problems, we may interpret $x_i$ as being the index of each edge, and $y_i$ as being the number $\tau(e_i)$ which is always less than the index itself; similar observations apply to the other problems.
\item
Each query on a slice $G_{i,j}$ is translated to dominance queries determined by the point $(i,j)$. The number of edges in the slice, $j-i+1$, is proportional to the geometric distance $(j-i)\sqrt 2$ of this point from the main diagonal.
\end{itemize}

We may assume without loss of generality that the quadrant in which we wish to count points for a query $(i,j)$ is the quadrant $\{(x,y)\mid x\le i \wedge y\ge j\}$ that extends from the query point towards the main diagonal. It is not true that these are the only quadrants produced by our reductions from graph slice problems to dominance counting; however, the number of points in each of the other three quadrants may be easily computed by combining the number of points in this quadrant with halfspace range counting problems. The number of points in an axis-aligned halfspace can be determined trivially in linear space and constant time per query by precomputing the answer to each possible query halfspace.
Thus, it remains to show that, given any set of points below the main diagonal of the square, we can answer dominance counting problems for quadrants that point towards the main diagonal, in an amount of time per query that is a function of the distance from the diagonal.

To solve dominance counting problems on a given set of points, satisfying the assumptions, we partition the points into subsets, where subset $S_i$ contains the points whose distance from the main diagonal is at most $(\log n)^{2^i}$ and which are not in any set $S_{i'}$ for $i'<i$.
Then, in outline, we use a local coordinate system for each subset $S_i$ in which the number of distinct coordinates is proportional to the number of points in $S_i$ (allowing the data structure of JaJa et al. to be used in a space-efficient way) and we cover each subset $S_i$ by data structures that each serve a range of $O((\log n)^{2^i})$ coordinates, in such a way that each point is covered by at most two data structures; again, this achieves linear space, while allowing a query within $S_i$ to be performed quickly. Finally, we use \emph{fractional cascading}~\cite{ChaGui-Algo-86} to link each subset $S_i$ to the next subset $S_j$, allowing the transformation into the local coordinate systems to be performed quickly and allowing us to quickly find the subset $S_i$ in which it is most appropriate to perform the query, reducing all lower-level queries to constant-time halfspace counting queries.

In more detail, we store the following for each subset $S_i$:
\begin{itemize}
\item Lists of the points in $S_i$, sorted both by their $x$-coordinates and by their $y$-coordinates.
\item For each point in $S_i$, its indices in both sorted lists, allowing us to answer in constant time a halfspace counting query with the coordinate of that point.
\item Two lists $X_i$ and $Y_i$, consisting both of points in $S_i$ and of some points in $S_{j}$ for $j>i$, sorted by their $x$-coordinates and $y$-coordinates respectively. $X_i$ consists of $S_i$ together with the elements at even positions in $X_{i+1}$, and similarly $Y_i$ consists of $S_i$ together with the elements at even positions in $Y_{i+1}$. Each entry in $X_i$ or $Y_i$ contains pointers to the nearest point in the sorted list for $S_i$ and to the nearest point in $X_{i+1}$ or $Y_{i+1}$. In this way, starting from $S_0$, we can navigate from $S_i$ to $S_{i+1}$ in constant time.
\item For each point in $S_i$, a translation of its coordinates into the local coordinate system of $S_i$, obtained by compressing out coordinate values that occur neither as the $x$-coordinate nor as the $y$-coordinate of any point in $S_i$. In this compressed coordinate system, all points remain below the main diagonal, and the number of distinct coordinates is at most equal to the number of points.
\item A sequence of the data structures of JaJa et al., each covering (for some integer $k$) the subset of points in $S_i$ whose local coordinates have $y\ge k(\log n)^{2^i}$ and $x\le(k+2)(\log n)^{2^i}$. Thus, there are at most $2(\log n)^{2^i}$ distinct $x$- and $y$-coordinates within one of these structures, so their query time is
$$
\log\left( 2(\log n)^{2^i} \right) / \log\log n = O(2^i).
$$
Any query defined by a point $(i,j)$ with $j-i\le (\log n)^{2^i}$ may be handled by one of these structures, determined in constant time by dividing the query coordinates by $(\log n)^{2^i}$. Each point of $S_i$ belongs to two of these structures, so the total space for all of these structures is $O(|S_i|)$.
\end{itemize}
In addition, we store an array indexed by coordinate, mapping coordinates in the coordinate space of the whole point set to their positions in lists $X_0$ and $Y_0$.

\begin{theorem}
Given a set of $O(n)$ points below the main diagonal in an $n\times n$ integer grid, we can process them into a data structure of size $O(n)$ that handles dominance queries for which the query point is at distance $d$ from the main diagonal in time $O(\log d/\log\log n)$ per query.
\end{theorem}

\begin{proof}
All of the data structures described above take space $O(|S_i|)$ for each set $S_i$, so the total space is linear.

To answer a query, we start in $S_0$. Within each set $S_i$ for which the query quadrant extends beyond the distance of the set from the main diagonal and therefore could also contain points of $S_{i+1}$, we translate the query into two halfspace queries, answer these queries in constant time, and use the $X_i$ and $Y_i$ structure to progress to the next set $S_{i+1}$ in constant time. In the final set $S_i$, we translate the query into the local coordinate system and then use one of the data structures of JaJa et al.{} stored for this set to answer the query directly in time $O(2^i)$. This $O(2^i)$ time dominates the query (everything else is $O(i)$) and thus the time per query is $O(2^i)=O(\log d/\log\log n)$.
\end{proof}

When translated to our relational event graph problems, this gives query time bounds of the form $O(\log(j-i)/\log\log m)$ for querying slice $G_{i,j}$ of a relational event graph with $m$ edges.

We observe that the same improvement may also be applied to the one-dimensional colored range counting problem considered by Gupta et al.~\cite{ColorRangeQueries}: as in our results, Gupta et al.{} transform the given input into a range counting problem on a set of two-dimensional points below the main diagonal. They use three-sided range queries rather than dominance counting, but their queries may be replaced by a linear combination of two axis-aligned halfspace queries and a dominance query. And, as in our problems, the length of the query interval for colored range counting translates into the distance of the dominance query point from the main diagonal.

\section{Simplified dominance counting}
Our data structure for range searching uses fractional cascading layered on top of
multiple copies of the structure of JaJa, Mortensen and Shi~\cite{JaJMorShi-ISAAC-04}, which itself is quite complex and in turn relies on the fusion trees of Fredman and Willard~\cite{FreWil-JCSS-93}, which are also complex. Therefore, although it achieves a good asymptotic space and query time complexity, we do not expect this combination of methods to be easy to implement. In this section we outline an alternative data structure for the same dominance counting problems that we expect to be more practical, although its time and space bounds are larger and we have not tested its practicality. Additionally, compared to the data structure in the previous appendix, the structure we define in this appendix has the theoretical advantage that it can handle queries with weighted points (dominance sum queries) and not just queries with unweighted points (dominance counting queries).

\begin{figure}[t]
\centering\includegraphics[width=4in]{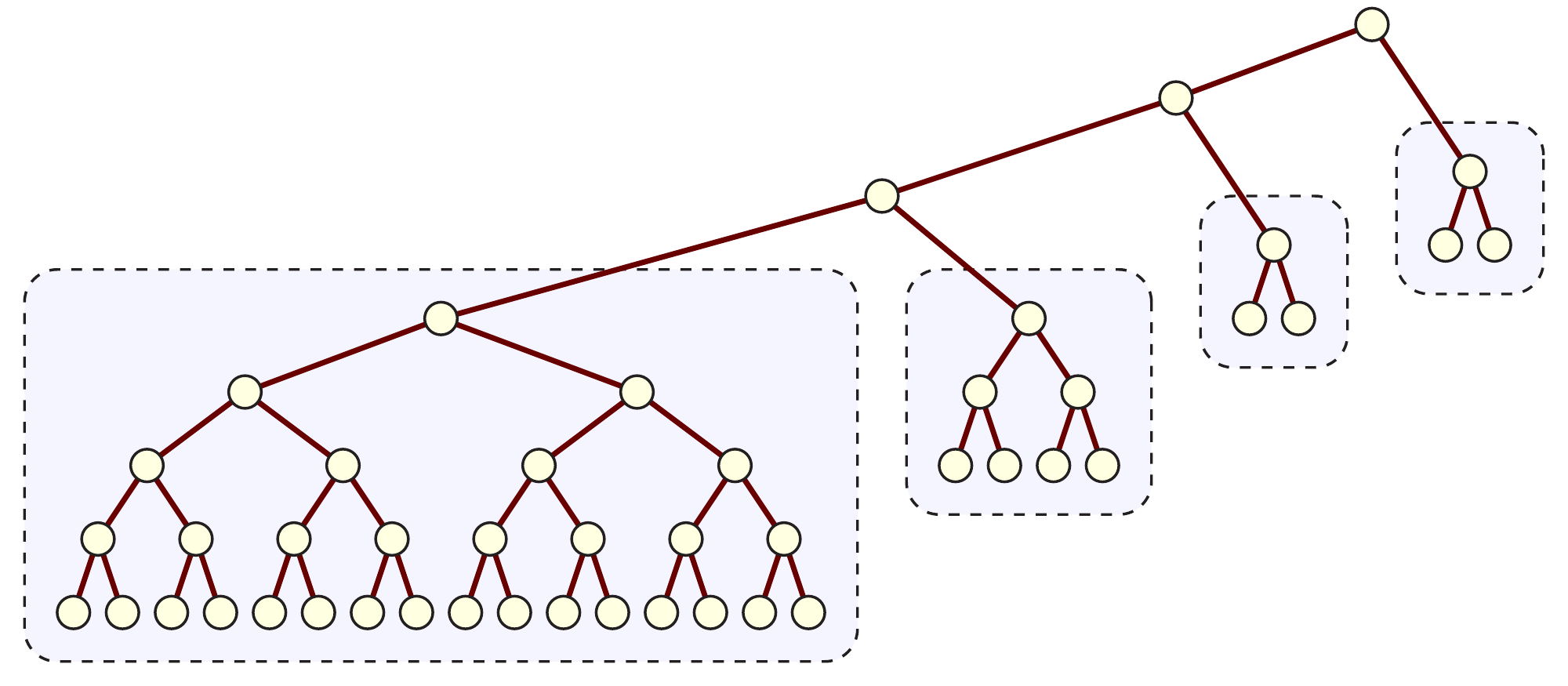}
\caption{24-leaf binary tree formed from the zeroless binary representation $24_{10}=2112_2$. Each of the four shaded complete binary subtrees corresponds to one of the four digits of the binary representation, in left-to-right order.}
\label{fig:foster-tree-24}
\end{figure}

In outline, our data structure for this problem uses path-copying persistence~\cite{DSPersistent} applied to a form of balanced binary tree, optimized for queries on small slices.
The specific trees we use are based on the observation that every positive integer has a unique representation as a base-2 number in which each digit is either 1 or 2 (rather than the more traditional binary notation in which each digit is either 0 or 1).\footnote{For an analogous representation of positive integers in base~10 using digits with values from 1 to 10, without a zero digit, see Foster~\cite{Fos-MM-47}.} For instance,
$$24_{10}=2112_2=2\times 2^3+1\times 2^2+1\times 2^1+2\times 2^0.$$
Based on this fact, for every $n$ we can form a tree $T_n$ with exactly $n$ leaves and $n-1$ internal vertices: we represent $n$ as $n=\sum_{i=0}^k b_i 2^i$ where each $b_i\in\{1,2\}$ and $k-1$ is the number of digits in the representation of $n$. We form a tree starting from a path of $k$ nodes, extending leftwards from the root; the right child of the node in this path at distance $i$ from the root is a complete binary tree with $b_i2^i$ leaves, and the left child of the last node in this path (at distance $k-1$ from the root) is a complete binary tree with $b_k2^k$ leaves. Figure~\ref{fig:foster-tree-24} illustrates this construction for $n=24$.

In $T_n$, the path from the root to the $i$th leaf (in the left-to-right ordering of the leaves) has length $O(\log(n-i))$: it takes at most $\log_2(n-i)$ steps to reach the complete binary subtree containing the $i$th leaf, and another $\log_2(n-i)+O(1)$ steps to reach the leaf from the root of this subtree. In addition, the structural change needed to form $T_{n+1}$ from $T_n$ is small: the binary representation of $n+1$ may be obtained from the representation of $n$ by changing trailing 2's to 1's and incrementing the lowest order digit that is not a 2, and each of these operations corresponds to $O(1)$ changes to the structure of the tree. So, in the worst case, $T_n$ and $T_{n+1}$ differ in the connections of $O(\log n)$ of their nodes, and the average change per step in constructing $T_n$ from $T_1$ by a sequence of these increment steps is $O(1)$. In particular, $T(n)$ can be constructed in time $O(n)$.

We now describe how to use these trees to solve dominance range sum queries. We assume we are given as input a set of $n$ points $(x_i,y_i)$, each with a weight~$w_i$. As in the previous section, we assume that $0\le y_i\le x_i<n$, so all points are on or below the main diagonal of the $n\times n$ integer grid. We wish to handle queries that are given as arguments a pair of coordinates $(x,y)$ and that return the query value
$$Q(x,y)=\sum\{w_i\mid x_i< x\wedge y_i> y\}.$$
That is, we sum the weights of the points in the quadrant of the plane directed towards the main diagonal from the query point.

To do so, for each value of $x$ in the range from $1$ to $n-1$ we store a tree $T_x$ with the structure described above, with exactly $x$ leaves. We represent each interior node of this tree as an object~$x$, with four instance variables: a weight $x.w$, a count $x.c$, and left and right child pointers $x.l$ and $x.r$. We do not explicitly represent the leaf nodes of the tree, but they are useful for defining its structure. The count variable for each node stores the number of leaves in the right subtree beneath that node; it is zero for nodes that are themselves leaves. Although leaves are not represented explicitly within our structure, we nevertheless define the weight of the leaf in position $i$ (in the left to right order of the leaves, starting from position~0 for the leftmost leaf) to be
$$\sum\{w_i\mid x_i<x\wedge y_i=y\}.$$
That is, it is the sum of weights of points within row $y$ of the $n\times n$ grid, up to column $x$.
The weight of a node that is not a leaf is the sum of the weights of the leaves in its right subtree.

\begin{figure}[t]
\centering\includegraphics[width=4.5in]{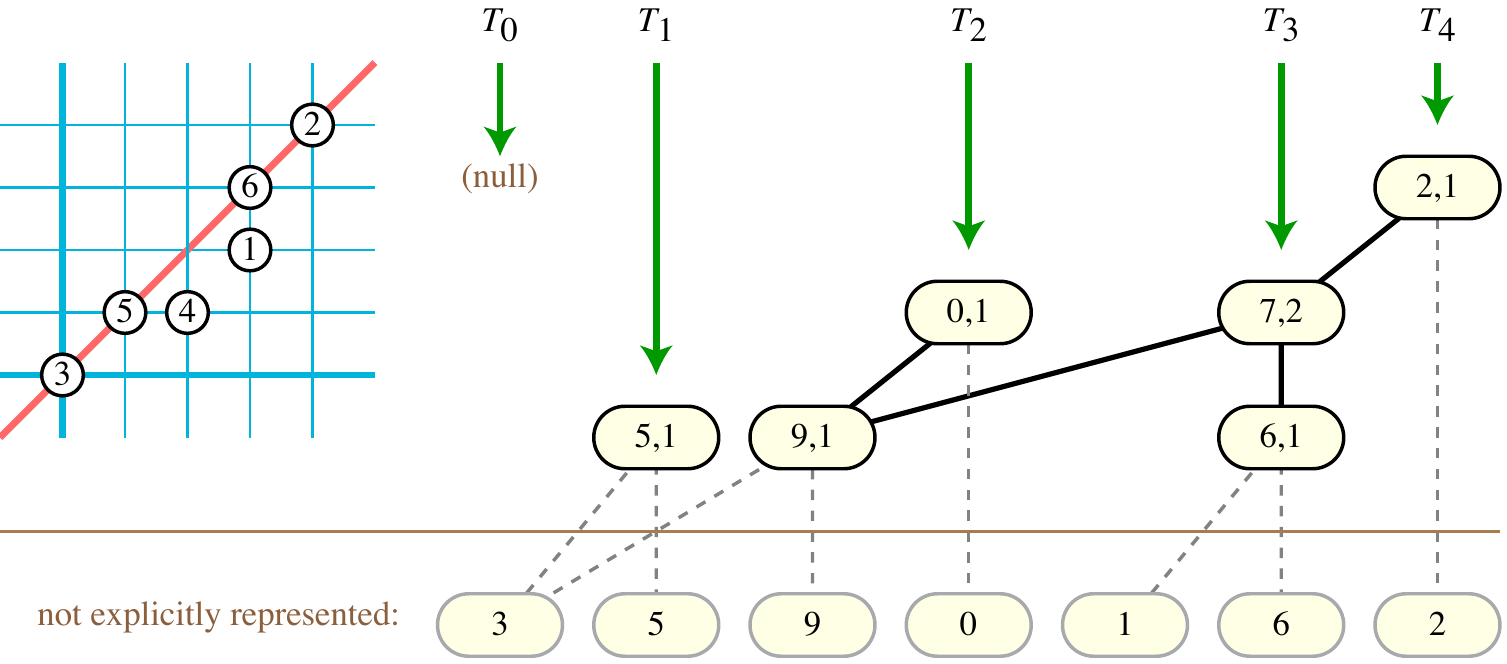}
\caption{The trees $T_n$ for $0\le n<5$ derived from a $5\times 5$ grid, and the pairs $x.w,x.c$ stored with each internal node of each tree. The leaf nodes are shown in the figure but not explicitly represented.}
\label{fig:shared-trees}
\end{figure}

To save space, we make the trees $T_x$ for different values of $x$ share as much of their structure as they can. In particular, if trees $T_x$ and $T_{x+1}$ both contain nodes whose descendants form isomorphic subtrees, with leaves in the same positions in the left-to-right order and with the same leaf weights, then our data structure reuses the same node object for both of them. However, the parts of $T_x$ and $T_{x+1}$ that differ either structurally or in the weights stored in those nodes are represented in the data structure by separate nodes. Finally, for each $x$ we store a pointer to the root of $T_x$ and we store a number $W_x$, the total weight of all the leaves in $T_x$. Figure~\ref{fig:shared-trees} illustrates this structure of shared trees, weights, and counts for a set of points within a $5\times 5$ grid.

To answer a query $(x,y)$, we perform a binary search for $y$ in tree $T_x$, using the count values stored in each tree node to guide whether to step leftwards or rightwards at each point of the search. The query value is then the sum of the weight values of the nodes at which this search stepped leftwards. As a special case, the query $(x,-1)$ (which asks for the sum of weights of all points with $x_i<x$) is handled by returning $W_x$. Thus, each query can be answered in time $O(\log(x-y))$.

To construct $T_x$ from $T_{x-1}$, we perform the structural rearrangements needed to form $T_x$ (creating new node objects for the root of each subtree in $T_x$ that does not also appear as a subtree in $T_{x-1}$. Then, for each point $(x_i,y_i)$ with $x_i=x-1$, we add $w_i$ to the weight value of leaf $y_i$ in $T_x$, and update the cumulative weights stored at each ancestor of this leaf, creating new copies of each ancestor node in order to be able to store these updated weight values in $T_x$ without disturbing the values already computed for $T_{x-1}$. The total number of new nodes that need to be created in this step for each point $(x_i,y_i)$ is $O(\log(x_i-y_i))$.
Thus, the total number of new nodes needed to create the entire structure, which gives the space requirement for the structure as well as its construction time, is $O(n+\sum_i \log(x_i-y_i))$.

We have proved the following result:

\begin{theorem}
Suppose we are given $m$ weighted points below the main diagonal in an $n\times n$ grid. Then in time $O(n+m\log n)$ we may preprocess these points into a data structure of size $O(n+m\log n)$ that supports dominance sum queries, given by a query point $(x,y)$, in time $O(\log(x-y))$ per query.
\end{theorem}

\section{Lower bounds}
\label{sec:lb}

Strengthening earlier results of Chazelle~\cite{Cha-JACM-90},
P{\v a}tra{\c s}cu provided lower bounds for two-dimensional range counting~\cite{Pat-STOC-07} that we adapt to our windowed relational event problems. Specifically, he showed that, for $n$ given points in the Euclidean plane, it is hard to answer \emph{dominance queries}, asking for the number of given points $(x_i,y_i)$ with $x_i\le X$ and $y_i\le Y$ for some query pair $(X,Y)$. In the cell probe model of Fredman and Saks~\cite{FreSak-STOC-89}, with $O(\log n)$ bits per machine word, every data structure that can answering such queries using space $O(n\log^{O(1)} n)$ requires $\Omega(\log n/\log\log n)$ time per query.

We adapt this lower bound to our windowed relational event problems, by showing how to translate a given set of $n$ Euclidean points into a synthetic relational event data set in such a way that windowed queries into this data set simulate range counting queries. To do so, we construct a (static) 1-regular graph with $2n$ vertices and $n$ isolated edges $e_i$. We may assume without loss of generality (by perturbing the Euclidean points if necessary) that no two points have the same $x$- or $y$-coordinate as each other; since only the ordering of the points by their coordinates matters for handling dominance queries, we may also assume (as P{\v a}tra{\c s}cu does) that their coordinates are all integers in the range from $0$ to $n-1$. That is, we are assuming that there exists a permutation $\pi$ of the integers from $0$ to $n-1$ such that the points in the given set of points all have coordinates of the form $(i,\pi(i))$.

Given a point set in this form, we define a relational event data set with $2n$ events, where for each $i$ in the range from $0$ to $n-1$, we include two copies of edge $e_i$ in the data set, one at time $n-i-1$ and a second copy at time $n+\pi(i)$. In this way, the number of given points dominated by the query pair $(X,Y)$ will exactly equal the number of repeated edges in the slice $G_{n-X-1,n+Y}$.

\begin{theorem}
\label{thm:lb}
For each of the problems of counting components, counting loopy components, counting isolated vertices, counting isolated edges, and counting repeated edges, any data structure for a relational event graph with $m$ edges and space $O(m\log^{O(1)} m)$ requires $\Omega(\log m/\log\log m)$ time per query.
\end{theorem}

\begin{proof}
For the data set produced by our translation, the answer to any one of these queries can be combined with the (trivially calculated) number of edges in a slice to give the number of repeated edges within a slice. Using the translation described above, this could then be used to answer two-dimensional range counting queries in the same asymptotic query bound. Since range counting queries cannot be answered more quickly than the query time stated in the theorem, neither can these graph queries.
\end{proof}

\begin{figure}[t]
\centering
\includegraphics[scale=0.55]{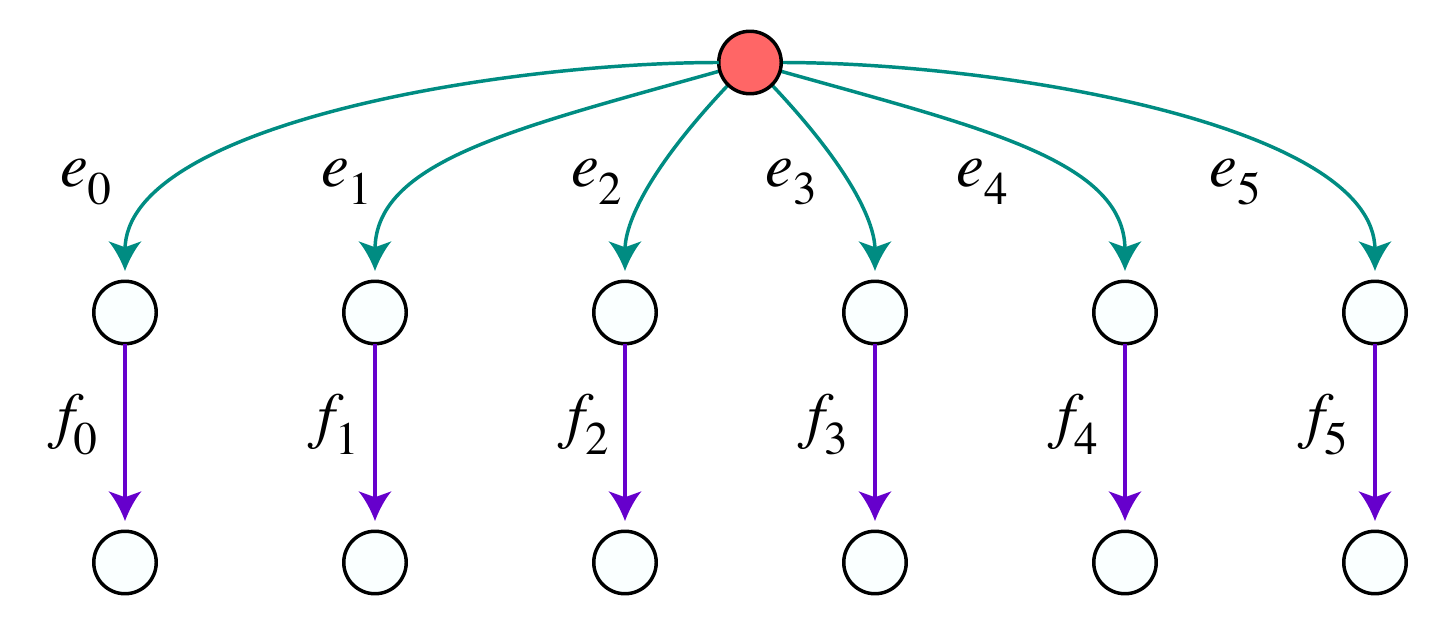}
\caption{Lower bound example for influenced vertices. The red vertex is influential; each remaining vertex is influenced only when the path leading to it is included in the slice.}
\label{fig:influence-lower-bound}
\end{figure}

A similar construction using a different relational event data set in the form of a tree of height two, using two-edge paths in place of the pairs of equal edges, shows that the same lower bound also holds for counting influenced vertices. Figure~\ref{fig:influence-lower-bound} shows the construction: for each pair of edges $(e_i,f_i)$, the endpoint of $e_i$ will be influenced whenever $e_i$ lies in the query interval, but the endpoint of $f_i$ will be influenced only if both $e_i$ and $f_i$ both lie in the query interval. Thus, as above, we can translate a two-dimensional range counting instance (represented as a permutation $\pi$ of the numbers from $0$ to $n-1$ by including a edge $e_i$ at time $n-i$ in a relational event data set and by including edge $f_i$ at time $n+\pi(i)$. We set the root of the tree as the sole influential vertex.

\begin{theorem}
\label{thm:lb2}
Any data structure for counting influenced vertices in a relational event graph with $m$ edges and space $O(m\log^{O(1)} m)$ requires $\Omega(\log m/\log\log m)$ time per query.
\end{theorem}

\begin{proof}
We use the translation described above. The answer to a dominance counting query for the query point $(X,Y)$ is given by $c-X$, where $c$ is the number of influenced vertices in the slice $G_{n-X-1,n+Y}$.
\end{proof}

\end{document}